\documentclass[11pt]{article}
\usepackage{amssymb,amsmath,cite,url}
\usepackage{float}
\usepackage{color,fullpage}
\usepackage{amsthm}
\newtheorem{theorem}{Theorem}
\newtheorem{corollary}{Corollary}
\newtheorem{lemma}{Lemma}

\newtheorem{assumption}{Assumption}

\begin{document}

\title{Strongly Convex Programming for Exact
Matrix Completion and Robust Principal Component Analysis}

\author{Hui Zhang\thanks{Department of Mathematics and Systems Science,
College of Science, National University of Defense Technology,
Changsha, Hunan, 410073, P.R.China. Corresponding author. Email: \texttt{h.zhang1984@163.com}}
\and Jian-Feng Cai\thanks{Department of Mathematics, University of Iowa, Iowa City, IA 52242, USA. Email: \texttt{jianfeng-cai@uiowa.edu}}
\and Lizhi Cheng\thanks{Department of Mathematics and Systems Science,
College of Science, National University of Defense Technology,
Changsha, Hunan, 410073, P.R.China.}
\and Jubo Zhu\thanks{Department of Mathematics and Systems Science,
College of Science, National University of Defense Technology,
Changsha, Hunan, 410073, P.R.China.}}


\date{}

\maketitle

\begin{abstract}
The common task in matrix completion (MC) and robust
principle component analysis (RPCA) is to recover a low-rank matrix
from a given data matrix. These problems gained great attention from various areas
in applied sciences recently, especially after the publication of the pioneering
works of Cand\`{e}s et al.. One fundamental result in MC and RPCA is
that nuclear norm based convex optimizations lead to the exact low-rank matrix
recovery under suitable conditions. In this paper, we extend this result by  showing that strongly convex optimizations can guarantee
the exact low-rank matrix recovery as well. The result in this paper not only
provides sufficient conditions under which the strongly convex models lead to the exact low-rank matrix recovery,
but also guides us on how to choose suitable parameters in practical algorithms.
\end{abstract}

%


\section{Introduction}\label{intro}

There is a rapidly growing interest in the recovery of an unknown
low-rank or approximately low-rank matrix from a sampling of its
entries. This problem occurs in many areas of applied
sciences ranging from machine learning \cite{aepNips2007} and control \cite{mpac1997}
to computer vision \cite{tkijcv1992}, and it has regained great attention since the publication
of the pioneering works \cite{cand1,cand4,rech1}.
Under different settings, low-rank matrix recovery problems can be
solved by nuclear norm based convex optimizations
without any loss of accuracy if the underlying matrix is indeed
of low rank \cite{cand1,cand2,cand3,cand4,rech1,rech2,rech3,rech4}. These fundamental results
have had a great impact in engineering and applied sciences.
Recently, they have shown their tremendous powers
in many practical applications such as
video restoration \cite{ji11,ji}, cognitive radio networks \cite{meng},
decomposing background topic from keywords \cite{min}, low-rank
textures capture \cite{zhan}, and 4D computed tomography \cite{gcsz}.



In order to design efficient algorithms for low-rank matrix recovery problems,
instead of directly solving the original convex optimizations,
we sometimes use their strongly convex approximations; see, e.g., \cite{cai2,wrig1}.
In this paper, we will show that these strongly convex programmings guarantee
the exact low-rank matrix recovery as well. Our result not only
provides sufficient conditions under which the strongly convex models lead to the exact low-rank matrix recovery, but also guide us on how to choose suitable parameters in practical algorithms.

\subsection{Convex programming}
In this paper, we mainly focus on two specific low-rank matrix recovery problems, namely, the low-rank matrix completion (MC) \cite{cand1,cand2,cai2} and the robust principle component analysis (RPCA) \cite{cand4,wrig1}.

Let $M\in\mathbb{R}^{n_1\times n_2}$ be an unknown low-rank matrix.
In MC, we would like to recover $M$ from
its partially known entries
$$
\{m_{ij},~~(i,j)\in\Omega\},
$$
where $\Omega\subsetneqq[n_1]\times[n_2]$ is
the set of indices of known entries. Here $[n]$ stands for the set $\{1,\ldots,n\}$.
Define $\mathcal{P}_{\Omega}$ as
\begin{equation}\label{eq2}
[\mathcal{P}_{\Omega}X]_{ij}=\left\{\begin{array}{lll} x_{ij},&
(i,j)\in \Omega,\\0,&  \mbox{otherwise}.
\end{array} \right.  \nonumber
\end{equation}
Then the sampled data in MC can be represented as $\mathcal{P}_{\Omega}M$.
The goal of MC is to recover the unknown low-rank matrix $M$ from
$\mathcal{P}_{\Omega}M$.
Since the matrix we are seeking is of low rank, a natural way to solve MC is to find the lowest
rank matrix among the feasible set
$\{X\in\mathbb{R}^{n_1\times n_2}:\mathcal{P}_{\Omega}X=\mathcal{P}_{\Omega}M\}$.
This leads to the following optimization problem
\begin{eqnarray}\label{eq3}
&\mathrm{minimize:}&~ ~\mathrm{rank}(X) \nonumber \\
&\mathrm{subject~ to:~
}&\mathcal{P}_{\Omega}X=\mathcal{P}_{\Omega}M.
\end{eqnarray}
Unfortunately, this problem is known to be NP-hard and the objective is non-convex.
A popular alternative is to relax (\ref{eq3}) to its nearest convex problem
\begin{eqnarray}\label{eq4}
&\mathrm{minimize:}&~ ~\|X\|_* \nonumber \\
&\mathrm{subject~ to:~
}&\mathcal{P}_{\Omega}X=\mathcal{P}_{\Omega}M.
\end{eqnarray}
Here $\|X\|_*$ is the summation of the singular values of $X$. Note that $\|X\|_*$ is the
best convex lower bound of the rank function on the set of matrices
whose top singular value is bounded by one \cite{faze}. It was
shown that, under suitable assumptions, the underlying low-rank matrix $M$
can be exactly recovered with high probability if $\Omega$ is uniformly randomly drawn from all
subsets of $[n_1]\times[n_2]$ with the same cardinality.

In RPCA, the sampled data matrix, denoted by
$D$, is $M$ with a small portion of its entries being corrupted. The
goal is to restore the low-rank matrix $M$ from its partially corrupted data matrix $D$.
Since only a small portion entries are corrupted, the noise matrix $D-M$ is a
sparse matrix. Therefore, $D$ consists of two components: a low-rank matrix
(the underlying matrix $M$) and a sparse matrix (the noise). Thus, the RPCA
problem can be solved by
\begin{eqnarray}\label{eq5}
&\mathrm{minimize:}&~ ~\mathrm{rank}(L)+ \lambda \|S\|_0 \nonumber \\
&\mathrm{subject~ to:~ }&D=L+S,
\end{eqnarray}
where $\|S\|_0$ stands for the number of nonzero entries of $S$.
This problem is NP-hard and non-convex. Again, one can exploit the convex relaxation
technique used in the MC problem. Notice that the convex lower bound
of the zero-norm function $\|\cdot\|_0$ in the infinity-norm unit ball
is the 1-norm $\|\cdot\|_1$ (the summation of absolute values).
Therefore, the convex relaxation of \eqref{eq5} is \cite{cand4,chan}
\begin{eqnarray}\label{eq6}
&\mathrm{minimize:}&~ ~\|L\|_*+\lambda \|S\|_1 \nonumber \\
&\mathrm{subject~ to:~ }&D=L+S.
\end{eqnarray}
where $\lambda$ is a  parameter balancing the low-rank and sparse
components. It was shown  in \cite{cand4,chan} that the convex optimization can recover the low-rank
matrix $M$ exactly under suitable assumptions.

\subsection{Strongly convex programming}
There exist many efficient algorithms
to solve the convex optimization problems \eqref{eq4} and \eqref{eq6} even when the scale of the problems is
large up to $10^5\times 10^5$, including \cite{cai1,liu,ma,toh,hui1}
for (\ref{eq4}) and \cite{cand4,lin1,lin2,wrig1} for (\ref{eq6}). However, instead of solving the convex minimizations \eqref{eq4} and \eqref{eq6} directly, some of the aforementioned algorithms solve their approximations involving
\emph{strongly convex} objectives.
In particular, the singular value thresholding (SVT) algorithm \cite{cai1} in
MC uses Uzawa's algorithm to solve
\begin{eqnarray}\label{eq7}
&\mathrm{minimize:}&~\|X\|_*+\frac{1}{2\tau}\|X\|^2_F \nonumber \\
&\mathrm{subject~ to:~
}&\mathcal{P}_{\Omega}X=\mathcal{P}_{\Omega}M,
\end{eqnarray}
and the iterative thresholding(IT) algorithm
\cite{wrig1} in RPCA solves\footnote{The paper \cite{wrig1} contains a critical error on the theoretical analysis, which had been discovered and removed by Emmanel Cand\`{e}s of Standford. Fortunately, the corrected version does not change the original convex model and has no effect on its iterative thresholding algorithm. See \url{http://books.nips.cc/papers/files/nips22/NIPS2009_0116_correction.pdf}.}
\begin{eqnarray}\label{eq8}
&\mathrm{minimize:}&~\|L\|_*+\frac{1}{2\tau}\|L\|^2_F+\lambda\|S\|_1+\frac{1}{2\tau}\|S\|^2_F \nonumber \\
&\mathrm{subject~ to:~ }&D=L+S.
\end{eqnarray}
where $\tau$ is some positive penalty parameter.
One of the main advantages of using  strongly convex programmings is that a broader range of existing optimization methods in the literature can be applied to the MC and RPCA problems. For example, it is well known that the convex conjugate of a strongly convex function is differentiable \cite{convex}. Therefore, the Lagrange dual of the strongly convex programmings \eqref{eq7} and \eqref{eq8} are differentiable, and hence smooth convex optimization methods can be applied to the dual. In fact, the SVT algorithm \cite{cai1} and the IT algorithm \cite{wrig1} are the gradient algorithms applied to the Lagrange dual of \eqref{eq7} and \eqref{eq8} respectively (a.k.a. Uzawa's algorithm). Additionally, one can exploit Nestrov's optimal scheme for smooth optimization \cite{nesterov} and even quasi-Newton method (e.g. L-BFGS \cite{yin1}), just to name a few.

However, the convex conjugate of $\|X\|_*$ or $\|X\|_1$ is not differentiable so that the Lagrange dual of \eqref{eq4} or \eqref{eq6} is not smooth. That gives us a part of reason why some algorithms exploit strongly convex approximations to solve \eqref{eq4} and \eqref{eq6}.
On the other hand, it follows from \cite{cai1,wrig1} and standard convex optimization theory
that, when $\tau$ tends to infinity, the strongly convex optimization
\eqref{eq7} and \eqref{eq8} become \eqref{eq4} and \eqref{eq6} respectively.
Therefore, in order to get the exact low-rank matrix recovery, we have to choose
an infinite $\tau$ in the SVT and IT algorithms, which is impractical.
Fortunately, it has been observed empirically in \cite{cai1,wrig1} that a finite $\tau$ is
enough for the purpose of the exact low-rank matrix recovery.
So a natural theoretical question is whether \eqref{eq7} and \eqref{eq8} with a finite $\tau$
lead to the exact low-rank matrix recovery in MC and RPCA.
Futhermore, the empirical convergence speed of the SVT and IT algorithms
depend on $\tau$. The smaller $\tau$ is, the faster the algorithm converges.
Therefore, it is  interesting to find a lower bound of $\tau$ for the exact low-rank matrix recovery.
These two questions are answered positively in this paper.

The related literature can be traced back to the linearized Bregman iteration
(LBI) algorithm \cite{cai2,cai3,cai4,oshe,yin1,yin2}, which can approximately
but efficiently solve the basis pursuit problem \cite{chen}
$\min_x\{\|x\|_1: Ax=b\}$ in compressed sensing
\cite{cand5,cand6,dono}. An interesting phenomenon about the LBI
algorithm was discovered in \cite{cai2,cai3,yin1}: the LBI algorithm converges to
a strongly convex optimization $\min_x\{\|x\|_1+\frac{1}{2\tau}\|x\|^2: Ax=b\}$
whose solution is the same as the basis pursuit problem when the parameter $\tau$ beyond
a finite scaler under some suitable conditions. In other words, under some suitable
conditions, a
convex programming is  equivalent to a strongly
convex programming when $\tau$ is large enough. In \cite{hui2}, this idea is extended to
the MC problem. Unfortunately, the bound  is very rough, and the method used
there has many limitations and can not be extended to the RPCA problem.
In \cite{frie}, a generic exact
regularization result is given, but the conditions are not easy to check for both MC and RPCA problems.

\subsection{Contributions and organization}
In the paper, we prove that, if $\tau$ exceed some value, then under some suitable conditions strongly convex programming \eqref{eq7}
and \eqref{eq8} can recover low-rank matrices exactly in MC and RPCA problems
respectively.  The explicit lower bounds of
$\tau$ are also given, and the lower bounds greatly improve the result in \cite{hui2}.

The significance of the paper is two-folded.
Firstly, sufficient conditions under which the strongly convex programming  leads to exact low-rank matrix recovery are derived. This, in turn,
allows us to exploit a broader range of existing optimization method in the literature.
Secondly, when we prefer to
minimizing the strongly convex objectives for designing fast algorithms,
we only need to set a finite parameter $\tau$ beyond some value
determined by the given data matrix, which lead to faster  convergence of the algorithms compared with the one using $\tau$ close to infinity.
In other words, it provides some
guidance on how to choose suitable parameters in practical algorithms.

The remainder of the paper is organized as follows. In Section 2, we
provide a brief summary of preliminaries including notations,
assumptions, and some existing results.
In Section 3, we show that  strongly convex programming \eqref{eq7}
and \eqref{eq8} lead to the exact low rank matrix recovery, and we give
explicit expressions of the lower bounds of $\tau$.
 Conclusion and further works are discussed in Section 4.

\section{Preliminaries} \label{sec:1}
In this section, we give some notations and existing results that will
be used later in this paper.

\subsection{Notations}\label{sec:2}
Let $X,Y\in\mathbb{R}^{n_1\times n_2}$ be two matrices.
The Euclidean inner product in
$\mathbb{R}^{n_1\times n_2}$ is defined as $\langle X,Y \rangle:=\mathrm{trace}(X^*Y)$,
where $X^*$ stands for the transpose $X$.
Then the Euclidean norm of $X$, denoted by $\|X\|_F$, is $\|X\|_F:=\sqrt{\langle X,X \rangle}$,
and it is also known as the Frobenius norm.
We denote the $i$-th nonzero singular value of matrix $X$ by
$\sigma_i(X)$, where $1\leq i\leq \mathrm{rank}(X)$.

The Forbenius norm also equals to the Euclidean norm of the vector of singular
values, i.e., $\|X\|_F=\left(\sum_{i}\sigma_i^2(X)\right)^{1/2}$. The
spectral norm of $X$, denoted by $\|X\|$, is $\|X\|:=\sigma_1(X)$, the largest
singular value of $X$. The nuclear norm of $X$
is the summation of its singular values, i.e.
$\|X\|_*:=\sum_{i}\sigma_i(X)$. The maximum entry of $X$ (in
absolute value) is denoted by $\|X\|_\infty :=\max_{i,j}|X_{ij}|$.
The $l_1$ norm of a matrix viewed as a long vector is denoted by
$\|X\|_1:=\sum_{i,j} |X_{ij}|$.

It can be easily verified that
\begin{eqnarray}\label{eq9}
\left\{\begin{array}{ll}
\|X\|_\infty \leq\|X\|_F\leq \|X\|_1 \leq n_1n_2\|X\|_\infty, \\
\|X\|\leq \|X\|_F\leq \|X\|_* \leq \sqrt{\mathrm{rank}(X)}\cdot\|X\|_F \leq \mathrm{rank}(X)\cdot\|X\|.
\end{array} \right.
\end{eqnarray}
The dual relationship is also important in our derivation. More precisely, we
have
\begin{eqnarray}\label{eq10}
\left\{\begin{array}{lll} &\|X\|=\sup_{\|Y\|_*\leq 1}\langle X,Y
\rangle, &\|X\|_*=\sup_{\|Y\|\leq 1}\langle X,Y
\rangle,\\
&\|X\|_1=\sup_{\|Y\|_\infty\leq 1}\langle X,Y \rangle,&
\|X\|_\infty=\sup_{\|Y\|_1\leq 1}\langle X,Y \rangle.
\end{array} \right.
\end{eqnarray}
It follows directly that $|\langle X,Y
\rangle|\leq \|X\|\cdot\|Y\|_*$ and $|\langle X,Y \rangle|\leq
\|X\|_1\cdot\|Y\|_\infty$.

We denote $M\in\mathbb{R}^{n_1\times n_2}$ to be the unknown low-rank
matrix that is to be recovered in the MC and RPCA problems. Recall
that the available data in the MC problem is $\mathcal{P}_{\Omega}M$, where
$\Omega$ is the index set of known entries. In the RPCA problem,
the available data is $D$. We will denote the sparse component, or the noise,
by $S_0=D-M$. Without ambiguity, $\Omega$ is used as the support of
$S_0$ in the RPCA problem.

Let $n=\max\{n_1,n_2\}$.
We will always use $r:=\mathrm{rank}(M)$, and denote the singular value decomposition of
$M$ by
$$
M=U\Sigma V^*=\sum_{i=1}^r\sigma_iu_iv_i^*,
$$
where
$\sigma_1,\dots, \sigma_r$ are the positive singular
values of $M$, and $U=[u_1,\dots, u_r]$ and $V=[v_1,\dots, v_r]$ are the
matrices of the left- and right-singular vectors.

Let $\mathrm{T}$ be the linear subspace of $\mathbb{R}^{n_1\times n_2}$
\begin{equation}\label{eq14}
\mathrm{T}:=\{UX_1^*+X_2V^*, X_1\in\mathbb{R}^{n_2\times r},
X_2\in\mathbb{R}^{n_1\times r}\},
\end{equation}
and $\mathrm{T}^\perp$ be its orthogonal completion. Define
$\mathcal{P}_{\mathrm{T}}$ and $\mathcal{P}_{\mathrm{T}^\perp}$ be the
orthogonal projections onto $\mathrm{T}$ and $\mathrm{T}^\perp$
respectively. Let $\mathcal{P}_{U}$ and $\mathcal{P}_{V}$
represent the projection operators onto the spaces spanned by the
columns of $U$ and $V$ respectively, i.e. $\mathcal{P}_{U}=UU^*$ and
$\mathcal{P}_{V}=VV^*$. The operator $\mathcal{P}_{\mathrm{T}}$
can be written as
\begin{equation}\label{eq15}
\mathcal{P}_{\mathrm{T}}Y=\mathcal{P}_UY+Y\mathcal{P}_V-\mathcal{P}_UY\mathcal{P}_V.
\end{equation}
To save notations,  $\Omega$ also denotes the linear space of the matrices supported on
$\Omega$, so that the projection onto $\Omega$ is $\mathcal{P}_{\Omega}$ in \eqref{eq2}.

It is well known that any subgradient \cite{convex} of the nuclear norm
function at $M$ is of the form
\begin{equation}\label{eq16}
UV^*+W,
\end{equation}
where $W\in\mathbb{R}^{n_1\times n_2}$ is a matrix satisfying
\begin{equation}\label{eq17}
\mathcal{P}_{\mathrm{T}}W=0, ~~~  \|W\|\leq 1.
\end{equation}
The subgradient of the $l_1$ norm function at $S_0$ is
\begin{equation}\label{eq18}
\mathrm{sgn}(S_0)+F,
\end{equation}
where $F$ satisfies $\mathcal{P}_{\Omega}F=0$ and $\|F\|_\infty \leq 1$,
and $\mathrm{sgn}(S_0)$ is the matrix whose entries are the signs of those of $S_0$.

Finally, we shall also manipulate some linear operators which act on
the space $\mathbb{R}^{n_1\times n_2}$. We will use calligraphic
letters for these operators as in $\mathcal{A}(X)$. In particular,
$\mathcal{I}:\mathbb{R}^{n_1\times n_2}\rightarrow
\mathbb{R}^{n_1\times n_2}$ denotes the identity operator and
$\mathcal{I}\preceq\mathcal{A}$ means that $\mathcal{A}-\mathcal{I}$
is symmetric positive semidefinite, and of course
$\mathcal{I}\prec\mathcal{A}$ means that $\mathcal{A}-\mathcal{I}$
is symmetric positive definite. We introduce the operator norm,
denoted by $\|\mathcal{A}\|$ and defined as $\|\mathcal{A}\|=\sup
_{\{\|X\|_F\leq1\}}\|\mathcal{A}X\|_F$.

\subsection{Existing results}
In this subsection, we present some existing results for exact low-rank
matrix recovery via convex optimization. For this, we need some assumptions
on the unknown low-rank matrix $M$ and the set $\Omega$.

Let us first give assumptions on $M$.
It is observed in \cite{cand1,cand3} that, for the MC problem,
it is impossible to recover a matrix which is equal to
zero in nearly all of its entries unless all of the matrix entries are observed. For the RPCA problem,
the authors in \cite{cand4,chan} observed if the low-rank
component is also sparse, then it is impossible to decide whether it
is low-rank or sparse. These observations lead to the introduction of the
incoherence assumptions that will be introduced in the following.
There are some other conditions to characterize the low-rank sparse decomposition,
e.g., the uncertainty principle and
rank-sparsity incoherence in \cite{chan}.

\begin{assumption}[Incoherence Assumption]
We say that the matrix $M$ obeys the incoherent assumption with parameter $\mu$,
if the following statements are valid:
\begin{equation}\label{eq11}
\max_i\|U^*e_i^{(1)}\|^2\leq \frac{\mu r}{n_1},~~~\max_i\|V^*e_i^{(2)}\|^2\leq
\frac{\mu r}{n_2}
\end{equation}
and
\begin{equation}\label{eq12}
\|UV^*\|_\infty \leq \sqrt{\frac{\mu r}{n_1n_2}}.
\end{equation}
where $e_i^{(1)}$ and $e_i^{(2)}$ are the $i$\textit{-th} vector in the canonical basis of
$\mathbb{R}^{n_1}$ and $\mathbb{R}^{n_2}$ respectively.
\end{assumption}

A stronger assumption, namely the strong incoherent condition, can also
guarantee exact matrix completion via convex programming \cite{cand2}.
It has been shown that the strong assumption implies the weaker one in
\cite{hui3}, but for the easy presentation of the existing results, we shall
introduce both of them.

\begin{assumption}[Strong Incoherence Assumption]
We say that the matrix $M$ obeys the strong incoherent assumption
with parameter $\mu_1$, if, for any $\mu\leq\mu_1$,
the inequality (\ref{eq12}) holds true and the following
statements are valid: for all pairs $(a,a^{'})\in[n_1]\times [n_1]$
and $(b,b^{'})\in[n_2]\times [n_2]$, it holds that
\begin{eqnarray}\label{eq13}
\left\{\begin{array}{lll} |\langle e_a^{(1)},UU^*e_{a^{'}}^{(1)}\rangle
-\frac{r}{n_1}1_{a=a^{'}}|\leq
\mu\frac{\sqrt{r}}{n_1};\\
|\langle e_b^{(2)},VV^*e_{b^{'}}^{(2)}\rangle -\frac{r}{n_2}1_{b=b^{'}}|\leq
\mu\frac{\sqrt{r}}{n_2}.
\end{array} \right.
\end{eqnarray}
\end{assumption}

Next we move to the assumptions on $\Omega$, the sampling pattern in
the MC problem and the sparse pattern of the sparse matrix in RPCA.
The theorems stated in \cite{cand1,cand4} always assume that $\Omega$
is uniformly randomly drawn from the set
$\{ \Lambda : |\Lambda|=m,\Lambda\subset[n_1]\times[n_2]\}$,
where $|\cdot|$ stands for the cardinality and $m=|\Omega|$.
However, as pointed out in \cite{cand4}, it is a little more
convenient to work with the Bernoulli model
$\Omega=\{(i,j):\delta_{i,j}=1\}$, where the $\delta_{i,j}$'s are
i.i.d. Bernoulli variable taking value one with probability $\rho$
and zero with probability $1-\rho$. Because our proofs will
heavily rely on the dual certificates constructed under the
Bernoulli model, we choose this proxy random model to the uniform
sampling, denoted by $\Omega\sim Ber(\rho)$, in our paper for
simplicity. We would like to point out that the convenience brought
by the Bernoulli model does not weaken the formal theorems; for more
details see \cite{cand4}.

Now we are ready to present two main results for exact matrix completion
and robust principle component analysis via convex optimization.

\begin{theorem}[{\cite[Theorem 1.2]{cand2}}]\label{the1}
Let $M\in \mathbb{R}^{n_1\times n_2}$ be a fixed matrix obeying the strong incoherence assumption with parameter
$\mu$. Let $n:=\max (n_1,n_2)$. Suppose we observe $m$ entries of
$M$ with locations sampled uniformly at random and there is a
positive numerical constant $C$ such that
\begin{equation}\label{eq19}
m\geq C\mu^2 nr\log^6n,
\end{equation}
Then $M$ is the unique solution to convex programming (\ref{eq4}) with
probability at least $1-n^{-3}$.
\end{theorem}

\begin{theorem}[{\cite[Theorem 1.1]{cand4}}]\label{the2}
Suppose $M\in \mathbb{R}^{n_1\times n_2}$ obeys the incoherence
assumption with parameter $\mu$, and that the support set of $S_0$
is uniformly distributed among all sets of cardinality $m$. Write
$n_{(1)}:=\max (n_1,n_2)$ and $n_{(2)}:=\min (n_1,n_2)$. Provided
that
\begin{equation}\label{eq20}
\mathrm{rank}(M)\leq \rho_r n_{(2)}\mu^{-1}(\log n_{(1)})^{-2},~~and ~~m\leq
\rho_sn_1n_2.
\end{equation}
where $\rho_r$ and $\rho_s$ are positive numerical constants. Then
there is a positive numerical constant $c$ with probability at least
$1-cn_{(1)}^{-10}$ such that $(M,S_0)$ is the unique solution to convex
programming (\ref{eq6}) with $\lambda=1/\sqrt{n_{(1)}}$.
\end{theorem}

\section{Main results}
In this section, we present the main results of this paper. More precisely,
we will show that, if $\tau$ is large enough, then the low rank matrix
$M$ can be exactly recovered by the strongly convex programming (\ref{eq7}) for
the MC problem and (\ref{eq8}) for the RPCA problem. A computable and explicit
lower bound of $\tau$ will  be given.
The results are proved in terms of dual
certificates (also known as Lagrangian multiplier) in combination with
a relaxation method introduced in \cite{gros} and studied
in \cite{cand4}.

\subsection{Matrix completion}
In this section, we discuss the strongly convex programming for the matrix
completion problem. Recall that, in the MC problem, we would like to recover
an unknown low-rank matrix $M$ from the sampling of its partial entries
$\mathcal{P}_{\Omega}M$. We will show that $M$ is the unique solution
of the strongly convex minimization \eqref{eq7} with
dominant probability, provided that $\tau$ exceeds than a specified finite number.
We will also give an explicit lower bound of $\tau$ using only the sampling
$\mathcal{P}_\Omega M$.

We first give a sufficient condition for $M$ being the unique solution of (\ref{eq7}).
It follows from the standard convex optimization theory and \cite{cand2}.
\begin{theorem}\label{the3}
Assume that there is a matrix $Y$ that satisfies

(a) $Y =\mathcal{P}_\Omega Y$,

(b) $\mathcal{P}_\mathrm{T}Y=\frac{1}{\tau}M + UV^*$,

(c) $\|\mathcal{P}_{\mathrm{T}^\perp}Y\|\leq 1$.

Then $M$ is the unique solution of the strongly convex programming
(\ref{eq7}).
\end{theorem}

\begin{proof}
The Lagrangian function associated with the strongly convex programming (\ref{eq7})
is
\begin{equation}\label{eq21}
L(X,Y):=\|X\|_*+\frac{1}{2\tau}\|X\|_F^2-\langle
Y,\mathcal{P}_{\Omega}X-\mathcal{P}_{\Omega}M \rangle,
\end{equation}
where $Y$ is the the Lagrangian multiplier.
By the strong convexity of the objective and the first
order optimality of the Lagrangian function, $M$ is the unique
solution if and only if there is a dual matrix $Y$ satisfying
\begin{equation}\label{eq22}
0\in\partial\|M\|_*+\frac{1}{\tau}M-\mathcal{P}_{\Omega}Y.
\end{equation}
Since we only use the entries of $Y$ on $\Omega$ from the expression
above, it is more convenient to assume $Y=\mathcal{P}_{\Omega}Y$.
We have known that $\partial\|M\|_*=UV^*+W$ with
$\mathcal{P}_{\mathrm{T}}W=0$ and $\|W\|\leq 1$. Thus the sufficient
and necessary condition becomes: there exists $W$ such that
\begin{equation}\label{eq23}
W=Y-\frac{1}{\tau}M-UV^*
\end{equation}
with $Y=\mathcal{P}_{\Omega}Y$, $\mathcal{P}_{\mathrm{T}}W=0$, and
$\|W\|\leq 1$. Therefore, it suffices to find a matrix $Y$
obeying
\begin{eqnarray}\label{eq24}
\left\{\begin{array}{lll} &Y =\mathcal{P}_{\Omega}Y\\
&\mathcal{P}_\mathrm{T}(Y-\frac{1}{\tau}M - UV^*)=0 \\
&\|\mathcal{P}_{\mathrm{T}^\perp}(Y-\frac{1}{\tau}M - UV^*)\|\leq 1.
\end{array} \right.
\end{eqnarray}
By the expression of $\mathcal{P}_{\mathrm{T}}$ and the fact
$U^*U=V^*V=I_r$, it is not hard to see that
$\mathcal{P}_\mathrm{T}(UV^*)=UV^*$ and $\mathcal{P}_\mathrm{T}M=M$, and
hence $\mathcal{P}_{\mathrm{T}^\perp}(UV^*)=0$ and
$\mathcal{P}_{\mathrm{T}^\perp}M=0$. So \eqref{eq24} is equivalent to
\begin{eqnarray}\label{eq25}
\left\{\begin{array}{lll} &Y =\mathcal{P}_{\Omega}Y\\
&\mathcal{P}_\mathrm{T}Y=\frac{1}{\tau}M + UV^* \\
&\|\mathcal{P}_{\mathrm{T}^\perp}Y\|\leq 1,
\end{array} \right.
\end{eqnarray}
which concludes the proof.
\end{proof}

Next, we shall construct a valid dual certificate $Y$ obeying the
conditions listed in the theorem above. We need the following lemma from \cite{cand2}.
%
%

\begin{lemma}[{\cite[Corollary 3.7]{cand2}}, {\cite[Theorem 6]{cand3}}]\label{lem1}
Under the assumptions in Theorem \ref{the1}, we have, with probability at
least $1-n^{-3}$,
\begin{equation}\label{eq37.5}
\|\mathcal{P}_{\mathrm{T}^\perp}\Lambda\|\leq \frac{1}{2},\quad\mbox{where}
~\Lambda:=\mathcal{P}_{\Omega}\mathcal{P}_{\mathrm{T}}
(\mathcal{P}_{\mathrm{T}}\mathcal{P}_{\Omega}\mathcal{P}_{\mathrm{T}})^{-1}UV^*,
\end{equation}
and
\begin{equation}\label{eq38}
\frac{p}{2}\mathcal{I}\preceq \mathcal{P}_\mathrm{T}
\mathcal{P}_\Omega\mathcal{P}_\mathrm{T}\preceq
\frac{3p}{2}\mathcal{I},\quad\mbox{where}~p:=m/(n_1n_2).
\end{equation}
\end{lemma}

It is obvious that the matrix $\Lambda$ in Lemma \ref{lem1}
satisfies $\Lambda=\mathcal{P}_{\Omega}\Lambda$ and
$\mathcal{P}_{\mathrm{T}}\Lambda=UV^*$. Now we are going to construct a valid dual
certificate for exact matrix completion via strongly convex programming.

\begin{theorem}\label{the4}
Assume
\begin{equation}\label{eq36}
\tau\geq\frac{\|\mathcal{P}_{\mathrm{T}^\perp}\mathcal{P}_{\Omega}\mathcal{P}_{\mathrm{T}}
(\mathcal{P}_{\mathrm{T}}\mathcal{P}_{\Omega}\mathcal{P}_{\mathrm{T}})^{-1}M\|}{1-\|\mathcal{P}_{\mathrm{T}^\perp}\Lambda\|},
\end{equation}
where $\Lambda$ is defined in Lemma \ref{lem1}.
Then, under the assumptions in Theorem \ref{the1}, $M$ is the unique
solution to the strongly convex programming \eqref{eq7} with
probability at least $1-n^{-3}$.
\end{theorem}

\begin{proof}
We prove the theorem by constructing a matrix $Y$ which satisfies
conditions in Theorem \ref{the3} with high probability.
Indeed, if we define
$$
Y=\Lambda+\frac{1}{\tau}\mathcal{P}_{\Omega}\mathcal{P}_{\mathrm{T}}
(\mathcal{P}_{\mathrm{T}}\mathcal{P}_{\Omega}\mathcal{P}_{\mathrm{T}})^{-1}M
$$
with $\Lambda$ defined as in Lemma \ref{lem1}, then the
conditions in Theorem \ref{the3} are satisfied with
probability at least $1-n^{-3}$.

By Lemma \ref{lem1}, with probability at least $1-n^{-3}$,
we have $\|\mathcal{P}_{\mathrm{T}^\perp}\Lambda\|\leq \frac{1}{2}$.
We verify (a)(b)(c) in Theorem \ref{the3} under the events that
$\|\mathcal{P}_{\mathrm{T}^\perp}\Lambda\|\leq \frac{1}{2}$.
First of all, $\tau>0$. Therefore, \eqref{eq7} is a strongly convex
programming. Next, by the construction of $Y$,
it is obvious that (a) and (b) in
Theorem \ref{the3} are true. It remains to show (c). In fact,
\eqref{eq36} implies that
$$
\frac{1}{\tau}\|\mathcal{P}_{\mathrm{T}^\perp}\mathcal{P}_{\Omega}\mathcal{P}_{\mathrm{T}}
(\mathcal{P}_{\mathrm{T}}\mathcal{P}_{\Omega}\mathcal{P}_{\mathrm{T}})^{-1}M\|\leq
1-\|\mathcal{P}_{\mathrm{T}^\perp}\Lambda\|.
$$
Since $\|\mathcal{P}_{\Omega^\perp}\Lambda\|\leq1/2$, we have
$$
\|\mathcal{P}_{\mathrm{T}^\perp}Y\|\leq \|\mathcal{P}_{\mathrm{T}^\perp}\Lambda\|+
\frac{1}{\tau}\|\mathcal{P}_{\mathrm{T}^\perp}\mathcal{P}_{\Omega}\mathcal{P}_{\mathrm{T}}
(\mathcal{P}_{\mathrm{T}}\mathcal{P}_{\Omega}\mathcal{P}_{\mathrm{T}})^{-1}M\|
\leq 1.
$$
Therefore, (c) is verified.
\end{proof}

The lower bound of $\tau$ in Theorem \ref{the4} is obviously finite.
However, it is hard to estimate in practice, as the quantities in \eqref{eq36}
are hard to estimate before we get $M$. In the following,
we present a computable lower bound of $\tau$, using only the given data $\mathcal{P}_{\Omega}M$.
Our starting point is  Lemma \ref{lem1}, which guarantee
\eqref{eq37.5} and \eqref{eq38} with very high probability.
Notice that
\begin{equation}
\| \mathcal{P}_{\Omega}\mathcal{P}_{\mathrm{T} } Z\|_F^2=
\langle  \mathcal{P}_{\Omega}\mathcal{P}_{\mathrm{T} }Z,  \mathcal{P}_{\Omega}\mathcal{P}_{\mathrm{T} }Z\rangle=
\langle Z, \mathcal{P}_{\mathrm{T} } \mathcal{P}_{\Omega}\mathcal{P}_{\mathrm{T} }Z\rangle.
\end{equation}
We have $$\| \mathcal{P}_{\Omega}\mathcal{P}_{\mathrm{T}}
(\mathcal{P}_{\mathrm{T}}
\mathcal{P}_{\Omega}\mathcal{P}_{\mathrm{T}})^{-1}M\|_F^2=\langle (\mathcal{P}_{\mathrm{T} } \mathcal{P}_{\Omega}\mathcal{P}_{\mathrm{T} })^{-1}M,M\rangle\leq \frac{2}{p}\|M\|_F^2,$$
and $$\| \mathcal{P}_{\Omega}\mathcal{P}_{\mathrm{T} } M\|_F^2\geq \frac{p}{2}\|M\|_F^2.$$
With these two inequalities, we can easily get
\begin{equation}\label{eq37}
\begin{split}
\frac{\|\mathcal{P}_{\mathrm{T}^\perp}\mathcal{P}_{\Omega}\mathcal{P}_{\mathrm{T}}
(\mathcal{P}_{\mathrm{T}}
\mathcal{P}_{\Omega}\mathcal{P}_{\mathrm{T}})^{-1}M\|}{1-\|\mathcal{P}_{\mathrm{T}^\perp}\Lambda\|}
&\leq
2\|\mathcal{P}_{\mathrm{T}^\perp}\mathcal{P}_{\Omega}\mathcal{P}_{\mathrm{T}}
(\mathcal{P}_{\mathrm{T}}
\mathcal{P}_{\Omega}\mathcal{P}_{\mathrm{T}})^{-1}M\|_F\leq 2\| \mathcal{P}_{\Omega}\mathcal{P}_{\mathrm{T}}(\mathcal{P}_{\mathrm{T}}
\mathcal{P}_{\Omega}\mathcal{P}_{\mathrm{T}})^{-1}M\|_F\cr
&\leq2 \sqrt{\frac{2}{p}} \|M\|_F
\leq \frac{4}{p}\|\mathcal{P}_{\Omega}\mathcal{P}_{\mathrm{T}}M\|_F=\frac{4}{p}\|\mathcal{P}_{\Omega}M\|_F.
\end{split}
\end{equation}
This bound can be computed by using only the sampling
$\mathcal{P}_{\Omega}M$. We summarize the
result into the following corollary.
\begin{corollary}\label{Cor1}
Assume
\begin{equation}
\tau\geq\frac{4}{p}\|\mathcal{P}_{\Omega}M\|_F,
\end{equation}
where $p=\frac{m}{n_1n_2}$ is the sampling ratio.
Then, under the assumptions in Theorem \ref{the1}, $M$ is the unique
solution to the strongly convex programming \eqref{eq7} with
probability at least $1-n^{-3}$.
\end{corollary}

\subsection{The Robust Principle Component Analysis problem}
Next, we discuss the RPCA case. Recall that, in the RPCA problem,
we would like to recover the unknown low-rank matrix $M$ from the
sampling $D$ of its complete entries with a small fraction of them
corrupted.
We denote the noise in $D$ by $S_0:=D-M$. Since
only a small fraction of entries are corrupted, $S_0$ is a sparse matrix.
Let $\Omega$ be the supported set of $S_0$. In this section, we shall
show that $(M,S_0)$ is the unique solution of the strongly convex programming
\eqref{eq8} with high probability, provided $\tau$ exceeding a finite number.
An explicit lower bound of $\tau$ will be given as well.

Similar to the MC case, we first present a theorem that states a sufficient
condition for $(M,S_0)$ being the unique solution of \eqref{eq8}.
\begin{theorem}\label{the5}
Assume $\|\mathcal{P}_{\Omega}\mathcal{P}_{\mathrm{T}}\|\leq 1/2$.
Suppose that there exists a pair $(W,F)$ and a matrix $B$ obeying
\begin{equation}\label{eq26}
UV^*+W+\frac{1}{\tau}M=\lambda
(\mathrm{sgn}(S_0)+F+\mathcal{P}_{\Omega}B)+\frac{1}{\tau}S_0
\end{equation}
with
$$
\mathcal{P}_{\mathrm{T}}W=0,~\|W\|\leq \beta,
\qquad \mathcal{P}_{\Omega}F=0,~\|F\|_\infty\leq \beta,
~\|\mathcal{P}_{\Omega}B\|_F\leq \alpha,
$$
where $\alpha, \beta$ are positive parameters satisfying
\begin{equation}\label{eq27}
\begin{cases}
\beta\leq 1, \cr
\alpha +\beta \leq1/\lambda,\cr
\lambda \leq(1-\beta)/2\alpha.
\end{cases}
\end{equation}
Then $(M,S_0)$ is the unique solution of the strongly convex programming
(\ref{eq8}).
\end{theorem}
\begin{proof}
The main idea of the proof follows the arguments of
\cite[Lemma 2.4]{cand4}. We consider a feasible perturbation of the
form $(M+H,S_0-H)$ and show that the objective increases unless
$H=0$. Let $UV^*+W_0+\frac{1}{\tau}M$ be an arbitrary subgradient
of $\|L\|_*+\frac{1}{2\tau}\|L\|^2_F$ at $M$, and so we have
$\mathcal{P}_{\mathrm{T}}W_0=0$ and $\|W_0\|\leq1$. Similarly, let
$\lambda(\mathrm{sgn}(S_0)+F_0)+\frac{1}{\tau}S_0$ be an arbitrary subgradient of
$\lambda\|S\|_1 +\frac{1}{2\tau}\|S\|^2_F$ at $S_0$, and we have
$\mathcal{P}_{\Omega}F_0=0$ and $\|F_0\|_{\infty}\leq1$.
By the definition of
subgradient, it holds
\begin{equation}\label{eq28}
\begin{split}
&f(M+H,S_0-H)\cr
&\geq f(M,S_0)+\langle UV^*+W_0+\frac{1}{\tau}M,H\rangle-
\langle \lambda(\mathrm{sgn}(S_0)+F_0)+\frac{1}{\tau}S_0,H\rangle.
\end{split}
\end{equation}
By the construction in \cite{cand4}, we can always choose $W_0$ in \eqref{eq28}
such that $\langle W_0,H\rangle=\|\mathcal{P}_{\mathrm{T}^\perp}H\|_* $, and $F_0$
such that $\langle F_0,H\rangle=-\|\mathcal{P}_{\Omega^\perp}H\|_1$.
Therefore,
\begin{equation}\label{eq29}
\begin{split}
&f(M+H,S_0-H)\cr
&\geq f(M,S_0)+\|\mathcal{P}_{\mathrm{T}^\perp}H\|_*+
\|\mathcal{P}_{\Omega^\perp}H\|_1+\langle UV^*+\frac{1}{\tau}M-\lambda
\mathrm{sgn}(S_0)-\frac{1}{\tau}S_0 ,H\rangle\cr
&= f(M,S_0)+\|\mathcal{P}_{\mathrm{T}^\perp}H\|_*+
\|\mathcal{P}_{\Omega^\perp}H\|_1\cr
&\qquad\qquad+\langle UV^*+\frac{1}{\tau}M-\lambda
\mathrm{sgn}(S_0)-\frac{1}{\tau}S_0 -\lambda\mathcal{P}_{\Omega}B ,H\rangle
+\lambda \langle \mathcal{P}_{\Omega}B,H\rangle\cr
&=f(M,S_0)+\|\mathcal{P}_{\mathrm{T}^\perp}H\|_*+
\|\mathcal{P}_{\Omega^\perp}H\|_1 +\langle -W+\lambda F ,H\rangle +\lambda \langle
\mathcal{P}_{\Omega}B,H\rangle,
\end{split}
\end{equation}
where in the last equality we have used the assumption (\ref{eq26}). By
the triangle inequality and the assumptions on $W$ and $F$, we derive
that
\begin{equation}\label{eq30}
\begin{split}
|\langle -W+\lambda F ,H\rangle|
&\leq |\langle W ,H\rangle|+\lambda |\langle  F ,H\rangle|\cr
&=|\langle W ,\mathcal{P}_{\mathrm{T}^\perp}H\rangle|+\lambda
|\langle  F ,\mathcal{P}_{\Omega^\perp}H\rangle|\cr
&\leq  \|W \|\cdot \|\mathcal{P}_{\mathrm{T}^\perp}H\|_*+\lambda \|
F \|_\infty \cdot\|\mathcal{P}_{\Omega^\perp}H\|_1\cr
&\leq \beta(\|\mathcal{P}_{\mathrm{T}^\perp}H\|_*+\lambda
\|\mathcal{P}_{\Omega^\perp}H\|_1).
\end{split}
\end{equation}
By the Cauchy-Schwarz inequality, it holds
\begin{equation}\label{eq31}
|\langle \mathcal{P}_{\Omega}B,H\rangle|=|\langle
\mathcal{P}_{\Omega}B,\mathcal{P}_{\Omega}H\rangle|
\leq \alpha\| \mathcal{P}_{\Omega}H\|_F.
\end{equation}
Furthermore, we have
\begin{equation}\label{eq32}
\begin{split}
\| \mathcal{P}_{\Omega}H\|_F&\leq
\|\mathcal{P}_{\Omega}\mathcal{P}_{\mathrm{T}}H\|_F+\|\mathcal{P}_{\Omega}\mathcal{P}_{\mathrm{T}^\perp}H\|_F
\leq \frac{1}{2}\|H\|_F+\|\mathcal{P}_{\mathrm{T}^\perp}H\|_F\cr
&\leq \frac{1}{2}\|\mathcal{P}_{\Omega}H\|_F+\frac{1}{2}
\|\mathcal{P}_{{\Omega}^\perp}H\|_F+\|\mathcal{P}_{\mathrm{T}^\perp}H\|_F,
\end{split}
\end{equation}
which implies
\begin{equation}\label{eq33}
\| \mathcal{P}_{\Omega}H\|_F\leq\|\mathcal{P}_{{\Omega}^\perp}H\|_F+2\|\mathcal{P}_{\mathrm{T}^\perp}H\|_F
\leq\|\mathcal{P}_{{\Omega}^\perp}H\|_1+2\|\mathcal{P}_{\mathrm{T}^\perp}H\|_*.
\end{equation}
Combining all the inequalities above together, we get
\begin{equation}\label{eq34}
\begin{split}
f(M+H,S_0-H)&-f(M,S_0)\cr
\geq&(1-\beta-2\alpha
\lambda)\|\mathcal{P}_{\mathrm{T}^\perp}H\|_*+(1-\lambda\beta-\lambda\alpha)
\|\mathcal{P}_{\Omega^\perp}H\|_1.
\end{split}
\end{equation}
This, together with \eqref{eq27}, implies that $(M,S_0)$ is
a solution to \eqref{eq8}. The uniqueness follows from the strong
convexity of the objective in \eqref{eq8}.
\end{proof}

Before we derive a lower bound of $\tau$ for the RPCA problem, we
need the following lemmas.
\begin{lemma}[{\cite[Lemmas 2.8]{cand4}}]\label{lem2}
Assume that $\Omega\sim Ber(\rho)$ with parameter $\rho \leq \rho_s$
for some $\rho_s>0$. Then there exists a matrix $W^L$ in the form of
$W^L=\mathcal{P}_{\mathrm{T}^\perp}Y$
satisfying: under the other assumptions of Theorem
\ref{the2},

(a)~$\|W^L\|<1/4$,

(b)~$\|\mathcal{P}_{\Omega}(UV^*+W^L)\|_F<\lambda/4$,

(c)~$\|\mathcal{P}_{\Omega^\perp}(UV^*+W^L)\|_\infty<\lambda/4$.
\end{lemma}

\begin{lemma}[{\cite[Lemmas 2.9]{cand4}}]\label{lem3}
Assume that $S_0$ is supported on a set $\Omega$ sampled as in Lemma \ref{lem2},
and that the signs of $S_0$ are i.i.d. symmetric (and
independent of $\Omega$). Then there exists a matrix $W^S$
in the form of $W^S=\mathcal{P}_{\mathrm{T}^\perp}Z$ satisfying:
under the other assumptions of
Theorem \ref{the2},

(a)~$\|W^S\|<1/4$,

(b)~$\|\mathcal{P}_{\Omega^\perp}W^S\|_\infty<\lambda/4$,
\end{lemma}

The construction of $W^L$ and $W^S$ can be found in \cite{cand4}.
Now, we show that $(M,S_0)$ is the unique solution to the strongly convex
programming \eqref{eq8} with high probability. The idea is to verify
that $W^L+W^S$ is a valid dual certificate when the parameter $\tau$
goes beyond some finite value.
More precisely, we check the conditions in Theorem \ref{the5}.

\begin{theorem}\label{the6}
Assume
\begin{equation}\label{eq40}
\tau\geq\frac{2\|\mathcal{P}_{\Omega^\perp}D\|_\infty+\lambda\|\mathcal{P}_{\Omega}(M-S_0)\|_F}
{\lambda(1-\lambda)}
\end{equation}
Then, under the assumptions of Theorem \ref{the2},  $(M,S_0)$ is
the unique solution to the strongly convex programming \eqref{eq8}
with probability at least $1-cn_{(1)}^{-10}$, where $c$ is the
numerical constant in Theorem \ref{the2}.
\end{theorem}

\begin{proof}
We prove the theorem by checking the conditions in Theorem \ref{the5}.
Due to the relation between $F$, $W$ and $B$ in \eqref{eq26},
it suffices to show that, there exists a matrix $W$ obeying
\begin{eqnarray}\label{eq35}
\left\{\begin{array}{lll} &\mathcal{P}_{\mathrm{T}}W=0,\\
&\|W\|\leq \beta, \\
&\|\mathcal{P}_{\Omega^\perp}(UV^*+W+
\frac{1}{\tau}M-\frac{1}{\tau}S_0)\|_\infty \leq \beta\lambda,\\
&\|\mathcal{P}_\Omega (UV^*+W -\lambda \mathrm{sgn}(S_0)+\frac{1}{\tau}M
-\frac{1}{\tau}S_0)\|_F\leq \alpha \lambda,
\end{array} \right.
\end{eqnarray}
where $\alpha,\beta$ satisfies \eqref{eq27}.

Let $W=W^L+W^S$ with $W^L$ and $W^S$ defined in Lemmas \ref{lem2}
and \ref{lem3} respectively. We show that \eqref{eq35} holds true.
For simplicity, we denote
$$
\gamma:=\|\mathcal{P}_{\Omega^\perp}(M-S_0)\|_\infty,\quad
\delta:=\|\mathcal{P}_{\Omega}(M-S_0)\|_F.
$$
We choose $\alpha=\frac{\epsilon}{2\lambda}$ and $\beta=1-\epsilon$,
where $0\leq\epsilon<\frac12$ is a parameter determined later.
Then, \eqref{eq27} is satisfied for any $\lambda\leq1$; note that the final choice $\lambda=1/\sqrt{n_{(1)}}\leq 1$.
By the expressions of $W^S$ and $W^L$, it is obvious that $W\in
\mathrm{T}^\perp$. Since $\|W^L\|<1/4$ and $\|W^S\|<1/4$, we have
$\|W\|\leq \|W^L\|+\|W^S\|<1/2\leq\beta$. For the third inequality in
(\ref{eq35}), we have
\begin{eqnarray}\label{eq41}
\begin{array}{lll}& &\|\mathcal{P}_{\Omega^\perp}(UV^*+W+
\frac{1}{\tau}M-\frac{1}{\tau}S_0)\|_\infty \\
&\leq&\|\mathcal{P}_{\Omega^\perp}(UV^*+W^L)\|_\infty+\|\mathcal{P}_{\Omega^\perp}W^S\|_\infty
+\frac{1}{\tau}\|\mathcal{P}_{\Omega^\perp}(M-S_0)\|_\infty \\
&\leq& \frac{\lambda}{4}+\frac{\lambda}{4}+\frac{1}{\tau}\|\mathcal{P}_{\Omega^\perp}(M-S_0)\|_\infty\\
&=&\frac{\lambda}{2}+\frac{\gamma}{\tau}.
\end{array}
\end{eqnarray}
For the last inequality in \eqref{eq35},
we notice that $\mathcal{P}_\Omega(W^S)=\lambda \mathrm{sgn}(S_0)$ as shown in \cite{cand4},
and therefore we have
\begin{eqnarray}\label{eq42}
\begin{array}{lll}
& &\|\mathcal{P}_\Omega (UV^*+W -\lambda \mathrm{sgn}(S_0)+\frac{1}{\tau}M
-\frac{1}{\tau}S_0)\|_F\\
&=&\|\mathcal{P}_\Omega (UV^*+W -W^S+\frac{1}{\tau}M
-\frac{1}{\tau}S_0)\|_F\\
&= & \|\mathcal{P}_\Omega (UV^*+W^L+\frac{1}{\tau}M
-\frac{1}{\tau}S_0)\|_F\\
&\leq & \|\mathcal{P}_\Omega
(UV^*+W^L)\|_F+\frac{1}{\tau}\|\mathcal{P}_\Omega(M
-S_0))\|_F\\
&=&\frac{\lambda}{4}+\frac{\delta}{\tau}.
\end{array}
\end{eqnarray}
In order that the last two inequalities in \eqref{eq35} are satisfied, we have to
choose a $\tau$ such that
$$
\frac{\lambda}{2}+\frac{\gamma}{\tau}\leq(1-\epsilon)\lambda,
\quad\mbox{and}\quad
\frac{\lambda}{4}+\frac{\delta}{\tau}\leq\frac{\epsilon}{2},
$$
which implies that
\begin{equation}\label{eq42.5}
\tau\geq\max\left\{\frac{\gamma}{(\frac12-\epsilon)\lambda},
\frac{\delta}{(\frac{\epsilon}2-\frac{\lambda}{4})}\right\}.
\end{equation}
To get a tighter lower bound of $\tau$, we choose $\epsilon$ to minimize
the maximum in the right hand side of \eqref{eq42.5}. Since
$\frac{\gamma}{(\frac12-\epsilon)\lambda}$ monotonically increases
and $\frac{\delta}{(\frac{\epsilon}2-\frac{\lambda}{4})}$ monotonically decreases
as $\epsilon$ increases from $\frac{\lambda}{2}$ to $\frac12$, it is easy to
see that the maximum is attained when
$\frac{\gamma}{(\frac12-\epsilon)\lambda}
=\frac{\delta}{(\frac{\epsilon}2-\frac{\lambda}{4})}$, which yields a solution
$\epsilon=\frac{\frac{\delta}{2}+\frac{\gamma}{4}}{\frac{\gamma}{2\lambda}+\delta}$
that always falls in $[\frac{\lambda}{2},\frac12]$ if $\lambda\leq1$.
Substituting it into \eqref{eq42.5}, we obtain an optimal lower bound of $\tau$
$$
\tau\geq\frac{2\gamma+\lambda\delta}{\lambda(1-\lambda)}.
$$
In viewing of $\mathcal{P}_{\Omega^\perp}S_0=0$ and therefore $\mathcal{P}_{\Omega^\perp}(M-S_0)=\mathcal{P}_{\Omega^\perp}D$, we get \eqref{eq40}.

In the above derivation, we have used only assertions (a)(b)(c) in Lemma \ref{lem2}
and (a)(b) in Lemma \ref{lem3}.
The probabilities that these assertions hold true and $\|\mathcal{P}_{\Omega}\mathcal{P}_{\mathrm{T}}\|\leq 1/2$ are shown in the proof
of \cite[Theorem 1.1]{cand4}, under the same assumptions of this theorem.
Therefore, the remaining of the proof
of this theorem follows directly from that of \cite[Theorem 1.1]{cand4}.

\end{proof}

It is easy to see that the lower bound of $\tau$ is a finite number.
However, the exact lower bound is very hard to get, because we can only manipulate the given
data matrix $D$. In the following, we give a method to estimate the
lower bound.
For this, we have to get upper bounds for the norms involved in
\eqref{eq40}. For the first norm, we simply use
$$
\|\mathcal{P}_{\Omega^\perp}D\|_\infty\leq\|D\|_{\infty}.
$$
For the second norm, by using the facts $\mathcal{P}_{\Omega^\perp}M=\mathcal{P}_{\Omega^\perp}D$ and $\|\mathcal{P}_{\Omega}\mathcal{P}_{\mathrm{T}}\|\leq 1/2$, we have
$$\|\mathcal{P}_{\Omega}M\|_F^2=\|\mathcal{P}_{\Omega}\mathcal{P}_{\mathrm{T}}M\|_F^2\leq \frac{1}{4}\|M\|^2_F=\frac{1}{4}(\|\mathcal{P}_{\Omega}M\|_F^2+\|\mathcal{P}_{\Omega^\perp}D\|_F^2),$$
which implies
$$
\|\mathcal{P}_{\Omega}M\|_F\leq \frac{\sqrt3}{3}\|\mathcal{P}_{\Omega^\perp}D\|_F.
$$
Therefore, we obtain
$$
\|\mathcal{P}_{\Omega}(M-S_0)\|_F=\|2\mathcal{P}_{\Omega}M-\mathcal{P}_{\Omega}D\|_F
\leq\|\mathcal{P}_{\Omega}D\|_F+2\|\mathcal{P}_{\Omega}M\|_F\leq\|\mathcal{P}_{\Omega}D\|_F+\frac{2\sqrt{3}}{3}\|\mathcal{P}_{\Omega^\perp}D\|_F.
$$
By using the Cauchy-Schwarz inequality, we finally have
$$
\|\mathcal{P}_{\Omega}(M-S_0)\|_F\leq \sqrt{1^2+\left(\frac{2\sqrt{3}}{3}\right)^2}\cdot\sqrt{\|\mathcal{P}_{\Omega}D\|_F^2+\|\mathcal{P}_{\Omega^\perp}D\|_F^2}=\frac{\sqrt{15}}{3}\|D\|_F.
$$
Therefore, in practice, we can choose
$$
\tau\geq\frac{2\|D\|_\infty+\frac{\lambda\sqrt{15}}{3}\|D\|_F}
{\lambda(1-\lambda)},
$$
 to guarantee
the exact sparse low-rank matrix decomposition by solving the strongly
convex optimization \eqref{eq8}. Similar to Corollary \ref{Cor1} for exact matrix completion, we have an analog lower bound estimate based on the observed data matrix  for exact low-rank and sparse matrices decomposition.
\begin{corollary}
Assume
$$
\tau\geq\frac{2\|D\|_\infty+\frac{\lambda\sqrt{15}}{3}\|D\|_F}
{\lambda(1-\lambda)}.
$$
Then, under the assumptions of Theorem \ref{the2},  $(M,S_0)$ is
the unique solution to the strongly convex programming \eqref{eq8}
with probability at least $1-cn_{(1)}^{-10}$, where $c$ is the
numerical constant in Theorem \ref{the2}.
\end{corollary}

\section{Conclusion}
In this paper, we have shown that \emph{strongly} convex optimizations
can lead to exact low-rank matrix recovery in both the matrix completion
problem and robust principle component analysis under suitable conditions. Explicit lower bounds for
the parameters $\tau$ involved are given. These results are complementary to the results in
\cite{cand1,cand2,cand4}, where convex optimizations are shown to lead to
the exact low-rank matrix recovery.


We would like to point out that the combination of the MC and RPCA
problems, i.e., matrix completion from grossly corrupted data
\cite{cand4} has been modeled as
\begin{eqnarray}\label{eq44}
&\mathrm{minimize:}&~ ~\|L\|_*+\lambda \|S\|_1 \nonumber \\
&\mathrm{subject~ to:~ }&\mathcal{P}_\Omega(L+S)=Y.
\end{eqnarray}
With the technique discussed in this paper, one can modify \eqref{eq44}
to its strongly convex counterpart that can guarantee the exact recovery as well.

Note that all models here assume that the observed data are exact. However,
in any real world application, one only have observed data corrupted
at least by a small amount of noise. Therefore, a possible
direction for further study is to consider the case of noisy data
\cite{cand3} in matrix completion and robust principle component
analysis.


\section*{Acknowledgement}
The work of Hui Zhang has been supported by the Chinese Scholar Council during his visit to Rice University. The work of Lizhi Cheng and Jubo Zhu has been supported by the National Science
Foundation of China (No.61072118 and No. 61002024).

\small{

}
\end{document}